\documentclass[12pt]{amsart}
\usepackage[cp1251]{inputenc}
\usepackage[english]{babel}
\usepackage{amsmath}
\usepackage{amssymb}
\usepackage{amsfonts}
\usepackage{epsfig}
\usepackage{srcltx}
\usepackage{cite}
\usepackage{float}
\usepackage[ a4paper,  mag=1000, includefoot,  left=2cm, right=2cm, top=2cm, bottom=2cm, headsep=1cm, footskip=1cm ]{geometry}
\usepackage[unicode,colorlinks]{hyperref}
\hypersetup{ colorlinks=true, citecolor=blue, linkcolor=blue}

\newtheorem{Th}{Theorem}
\newtheorem{Def}{Definition}

\newtheorem{Rem}{Remark}
\begin{document}

\thispagestyle{empty}

\title{White noise perturbation of locally stable dynamical systems}

\author{Oskar Sultanov}

\address{Oskar Sultanov,
\newline\hphantom{iii}  Institute of Mathematics USC RAS
\newline\hphantom{iii}  112 Chernyshevsky str., Ufa 450008, Russia}

\email{oasultanov@gmail.com}

\maketitle {\small
\begin{quote}

\noindent{\bf Mathematics Subject Classification: }{93E15, 34D10, 60H10}

\end{quote}
\begin{quote}
\noindent{\bf Abstract.}
The influence of small random perturbations on a deterministic dynamical system with a locally stable equilibrium is considered.
The perturbed system is described by the It\^{o} stochastic differential equation. 
It is assumed that the noise does not vanish at the equilibrium. 
In this case the trajectories of the stochastic system may leave any bounded domain with probability one. 
We analyze the stochastic stability of the equilibrium under a persistent perturbation by white noise on an asymptotically long time interval $0\leq t\leq \mathcal O(\mu^{-N})$, where $0<\mu\ll 1$ is a perturbation parameter.
\medskip

\noindent{\bf Keywords:} dynamical system, random perturbation, stability, stochastic differential equation
\end{quote} }

\section{Problem statement}
Consider the system of ordinary differential equations:
\begin{equation}
\frac{d{\bf x}}{d t}={\bf f}({\bf x}, t), \quad t\geq 0, \quad {\bf x}\in\mathbb R^n.
\label{GS}\end{equation}
Let ${\bf x}=0$ be an equilibrium, ${\bf f}(0, t)\equiv 0$.
Suppose the vector-valued function ${\bf f}({\bf x}, t)=(f_1({\bf x}, t),\dots, f_n({\bf x}, t))$ is continuous, satisfies a Lipschitz condition: $|{\bf f}({\bf x}, t)-{\bf f}({\bf z}, t)|\leq L |{\bf x}-{\bf z}|$ and a growth condition: $|{\bf f}({\bf x}, t)|\leq M(1+|\bf x|)$ for all ${\bf x}, {\bf z}\in\mathbb R^n$, $ t\geq 0$ with positive constants $L$, $M$.
Assume that there exists a local Lyapunov function $V({\bf x}, t)\in\mathcal C^{2,1}(\mathbb R^n\times\mathbb R)$ satisfying the following properties:
\begin{equation}
    \begin{array}{c}
        \displaystyle
            |{\bf x}|^2\leq V({\bf x}, t)\leq A |{\bf x}|^2,\quad |\partial_{{\bf x}} V|^2\leq B |{\bf x}|^2, \quad |\partial_{ x_i}\partial_{x_j} V|\leq C, \\
        \displaystyle
            \frac{dV}{d t}\Big|_{\eqref{GS}}\stackrel{def}{=}\frac{\partial V}{\partial  t}+\sum\limits_{k=1}^{n}\frac{\partial V}{\partial x_k} f_k\leq - \gamma V \quad  \forall \, |{\bf x}|\leq r_0, \quad t\geq 0,
    \end{array}
\label{LF}
\end{equation}
with constants $A,B,C,\gamma,r_0>0$.
This implies that the solution ${\bf x}(t)\equiv 0$ is locally asymptotically stable (see~\cite{Krasov63}).
In this case, system \eqref{GS} may have other stable fixed points. Note that such Lyapunov functions are constructed in the study of nonlinear differential equations (see, for instance,~\cite{OS10,LKOS13}). In this paper we investigate the stability of the trivial solution under a persistent perturbation by white noise.

Together with \eqref{GS} we consider the It\^{o} stochastic differential equation:
\begin{gather}\label{GSP}
d{\bf y}(t)={\bf f}({\bf y}, t)\, d t+\mu\, G({\bf y}, t)\,d{\bf w}({ t}), \quad  {\bf y}(0)={\bf y}_0\in\mathbb R^n.
\end{gather}
Here ${\bf w}( t)=(w_1( t),\dots,w_m( t))$ is $m$-dimensional Wiener process defined on a probability space $(\Omega,\mathcal F,{\bf P})$, $G({\bf y}, t)=\{g_{ij}({\bf y}, t)\}$ is a given continuous $n\times m$ matrix, satisfies the Lipschitz condition:
\begin{gather}\label{Lip}
    \|G({\bf y}, t)-G({\bf z}, t)\|\leq L |{\bf y} -{\bf z}| \quad \forall \, {\bf y}, {\bf z}\in\mathbb R^n, \ \ t\geq 0,
\end{gather}
and the growth condition:
\begin{gather}\label{Grow}
    \|G({\bf y}, t)\|\stackrel{def}{=}\max_{i,j}|g_{ij}({\bf y}, t)|\leq M(1+|{\bf y}|) \quad \forall\, {\bf y}\in\mathbb R^n, \ \   t\geq 0.
\end{gather}
A small positive parameter $0<\mu\ll 1$ controls the intensity of the random perturbation.
The matrix $G({\bf y}, t)$ and the initial data ${\bf y}_0$ are assumed to be deterministic.
Note that these constraints are sufficient for the existence and uniqueness of a global solution to the initial value problem \eqref{GSP} for all ${\bf y}_0\in\mathbb R^n$ (see~\cite[\S 5.2]{Oksendal98},\cite[\S 3.3]{Hasm69},\cite[\S 6.2]{Arnold74}).
In addition, suppose that the trivial solution ${\bf y}(t)\equiv 0$ is not preserved in the perturbed system, namely, $G(0,t)\not\equiv 0$.
The goal of this paper is to describe a class of matrices $G$ guaranteing the stability in probability of the trivial solution to system \eqref{GS} under a persistent perturbation by white noise.

Remind the concept of stability in probability~\cite[\S 5.3]{Hasm69},\cite{Kac98} or stochastic stability~\cite[Ch. 2]{Kushner},
\cite[Ch. 11]{Schuss10}:
\begin{Def}
The solution ${\bf x}( t)\equiv 0$ to system \eqref{GS} is said to be strongly stable in probability with respect to white noise uniformly for $G$ from the set $\mathcal P$, if
$$\forall\, \varepsilon, \nu > 0 \ \ \exists\, \delta, \Delta>0: \ \ \forall \, |{\bf y}_0|<\delta, \ \ \mu<\Delta, \ \ G\in\mathcal P$$
the solution ${\bf y}( t)$ to the initial value problem \eqref{GSP} satisfies the inequality:
$${\bf P}(\sup_{ t\geq 0}|{\bf y}( t)|\geq \varepsilon)\leq \nu.$$
\end{Def}
Thus, if the trivial solution is strongly stable, then the trajectory of the process ${\bf y}(t)$ with sufficiently small initial values and perturbations does not leave an arbitrary small neighbourhood of zero with probability tending to one.
If the solution ${\bf x}(t)\equiv 0$ is weakly stable\footnote{In this case it is assumed that ${\bf P}(|{\bf y}( t)|\geq \varepsilon)\leq \nu$ for all $t>0$.}, then almost all realizations may leave the ball $\mathcal B_\varepsilon=\{{\bf y}\in \mathbb R^n: |{\bf y}|<\varepsilon\}$, but at each instant $t=t_\ast$ the random variable ${\bf y}( t_\ast)$ satisfies $|{\bf y}( t_\ast)|<\varepsilon$ with probability tending to one.
It is easy to see that strong stability implies weak stability, but the converse is not true in general~\cite[\S 6.11]{Hasm69}.

Many authors (see~\cite{Hasm69,Kushner,Arnold74,Schuss10,Kac98,Mao94}) have considered the stability of stochastic differential equations with $G(0, t)\equiv 0$.
However, the problem of stability under persistent perturbations was investigated only in several works~\cite{Hasm69,Hasm65,VF,Hapaev88,LK14}. Let us briefly recall the main known results concerning the case $G(0, t)\not \equiv 0$.
Firstly, from~\cite{VF} it follows that there is no strong stability in autonomous systems (where both ${\bf f}$ and $G$ do not explicitly depend on $t$): the trajectories of perturbed system may leave any bounded domain with probability one.
In~\cite[\S 7.4]{Hasm69}, Khasminskii proved that if $\|G({\bf y}, t)\|$ decays sufficiently fast at infinity $t\to\infty$ and there exists a suitable global Lyapunov function, the solution ${\bf x}(t)\equiv 0$ remains stable in a certain sense.
On the other hand, the existence of a global Lyapunov function for stochastic equation \eqref{GSP} is sufficient for weak stability w.r.t. white noise with uniformly bounded matrix $G$ (see~\cite{Hasm65}). However, if the deterministic system has at least one more stable fixed point ${\bf x}_\ast\neq 0$, the solution ${\bf x}(t)\equiv 0$ is not even weakly stable w.r.t. such class of perturbations.
Thus, the stability under persistent perturbation by white noise occurs in a fairly narrow class of systems, and it is therefore worth while to consider the problem of stability on a finite time interval (see~\cite{Hapaev88,VF}).
One variant of this approach is to find the largest possible time interval $[0; T]$ on which solutions to the perturbed equation are close to the equilibrium of the deterministic system. It was shown in~\cite[Ch. 9]{Hapaev88} that if the matrix $G$ is uniformly bounded and there exists a local Lyapunov function having properties \eqref{LF} with $\gamma=0$, then the equilibrium is strongly stable on the interval $0\leq  t\leq \mathcal O(\mu^{-2})$. Similar results on weak stability were obtained in~\cite{LK14}, where a barrier for the Kolmogorov equation (see~\cite[\S 4.4]{Schuss10}) was constructed using a Lyapunov function known in a neighborhood of the equilibrium.
However, the stability of solution on longer time intervals remains an open question. The present paper is devoted to solving this problem.

In this work we prove that if $G$ is uniformly bounded matrix, then the locally asymptotically stable solution ${\bf x}(t)\equiv 0$ to the deterministic system is strongly stable with respect to white noise on an asymptotically long time interval $0\leq t\leq \mathcal O(\mu^{-2N})$ for any natural number $N\geq 1$.

\section{Stability on asymptotically long time interval}

Let us clarify the definition of stability we shall use in this section.
\begin{Def}
The solution ${\bf x}( t)\equiv 0$ to system \eqref{GS} is said to be strongly stable in probability with respect to white noise on a time interval $[0;  T]$ uniformly for $G$ from the set $\mathcal P$, if
$$\forall\, \varepsilon, \nu > 0 \ \ \exists\, \delta, \Delta>0: \ \ \forall \, |{\bf y}_0|<\delta, \ \ \mu<\Delta, \ \ G\in\mathcal P$$
the solution ${\bf y}( t)$ to the initial value problem \eqref{GSP} satisfies the inequality:
$${\bf P}(\sup_{ t\in[0; T]}|{\bf y}( t)|\geq \varepsilon)\leq \nu.$$
\end{Def}

We define the class $\mathcal A$ as a set of matrices $G({\bf y}, t)$ satisfying estimates \eqref{Lip}, \eqref{Grow}, and
$$
    \sup_{|{\bf y}|\leq r_0, t\geq 0}\|\sigma({\bf y}, t)\|<\infty, \quad \sigma\stackrel{def}{=}\frac{G\cdot G^\ast}{2}=\{\sigma_{ij}\}.
$$
Let $\mathcal A_h$ be a set of matrices $G$ from the set $\mathcal A$ such that $\|\sigma({\bf y}, t)\|\leq h$ for all $|{\bf y}|\leq r_0$, $t\geq 0$.
Let us also define the operator (see~\cite[\S 3.6]{Hasm69}):
$$
    \mathcal L\stackrel{def}{=}\partial_ t+\sum_{k=1}^n f_k({\bf y}, t) \partial_{x_k} + \mu^2 \sum_{i,j=1}^n \sigma_{ij}({\bf y}, t)\partial_{x_i}\partial_{x_j}.
$$

We have
\begin{Th}\label{Th1}
Suppose that unperturbed system \eqref{GS} has a Lyapunov function $V({\bf x},t)$ satisfying \eqref{LF}. Then for all $N\in\mathbb N$, $h>0$, and $0<\varkappa<1$ the solution ${\bf x}(t)\equiv 0$ to system \eqref{GS} is strongly stable in probability with respect to white noise on the time interval $[0; \mu^{-2N+\varkappa}]$ uniformly for $G\in\mathcal A_h$.
\end{Th}
\begin{proof}
It is reduced to constructing a suitable Lyapunov function for stochastic differential equations \eqref{GSP} on the basis of the Lyapunov function $V({\bf x},t)$ for unperturbed deterministic system \eqref{GS}.

Let $h>0$, $0<\varkappa<1$, $\nu>0$, and $0<\varepsilon<r_0$ be arbitrary positive constants,
${\bf y}(t)$  a solution to system \eqref{GSP} with $|{\bf y}_0|<\delta$, $G\in\mathcal A_h$, and $\tau_{\mathcal I}$ a first exit time of the solution ${\bf y}( t)$ from the domain
$$
    \mathcal I\stackrel{def}{=}\{({\bf y}, t)\in\mathbb R^{n+1}: |{\bf y}|< \varepsilon, \ \  0 <  t <  T\},\quad T>0.
$$
We define the function $s_t=\min\{ \tau_{\mathcal I}, t\}$. Then ${\bf y}(s_ t)$ is the process stopped at the first exit time from the domain $\mathcal I$.

Let us first consider the case $N=1$.
Here we reproduce the arguments of~\cite[Ch. 9]{Hapaev88} with some modifications.
The Lyapunov function for system \eqref{GSP} is constructed in the form:
\begin{gather}\label{V1}
    V_1({\bf y}, t; T)=V({\bf y}, t)+ \mu^2 n^2  h C \cdot (T- t),
\end{gather}
where $V({\bf y},t)$ is the Lyapunov function for system \eqref{GS}.
It is easy to see that $V_1\geq V\geq 0$ and
$$
    \mathcal L V_1  =   \frac{dV}{d t}\Big|_{\eqref{GS}}  +  \mu^2\sum_{i,j=1}^n \sigma_{ij} \, \partial_{x_i}\partial_{x_j} V - \mu^2 n^2  h C \leq -\gamma  V \leq 0
$$
for all $({\bf y}, t)\in \mathcal I$. Hence $V_1({\bf y}(s_ t),s_ t)$ is a nonnegative supermartingale~\cite[\S 5.2]{Hasm69}.
From the properties of the Lyapunov function $V$ and the definition of the function $s_t$, we get the following estimates:
\begin{equation}\label{Est}
\begin{array}{lll}
\displaystyle {\bf P}(\sup_{ 0\leq  t\leq  T} |{\bf y}( t)|\geq \varepsilon)
    & = & \displaystyle {\bf P}(\sup_{ 0\leq  t\leq  T} |{\bf y}( t)|^2\geq \varepsilon^2) \leq \\
    & \leq & \displaystyle {\bf P}(\sup_{0\leq  t\leq  T} V({\bf y}(t), t)\geq \varepsilon^2) \leq \\
    & \leq &  \displaystyle {\bf P}(\sup_{0\leq  t\leq  T} V_1({\bf y}(t), t; T)\geq \varepsilon^2) = \\
    & = & \displaystyle {\bf P}(\sup_{ t\geq 0} V_1({\bf y}(s_ t),s_t;T)\geq \varepsilon^2) \leq  \frac{  V_1({\bf y}_0,0;T)}{\varepsilon^2}.
\end{array}
\end{equation}
The last estimate follows from Doob's supermartingale inequality~\cite[\S 7.3]{VAD}. Note that the initial point ${\bf y}_0$ is assumed to be deterministic.
We define $T=\mu^{-2+\varkappa}$ and
$$
    \delta=\Big(\frac{\varepsilon^2 \nu}{2 A}\Big)^{1/2}, \quad \Delta=\Big(\frac{ \varepsilon^2 \nu }{2\, n^2 h \, C}\Big)^{1/\varkappa}.
$$
Then $V_1({\bf y}_0, 0;T)=A |{\bf y}_0|^2 +\mu^\varkappa n^2  h \, C \leq \varepsilon^2 \nu$ for all $|{\bf y}_0|<\delta$, $\mu<\Delta$.
Taking into account \eqref{Est}, we obtain the estimate:
\begin{gather}\label{EstP}
    {\bf P}(\sup_{0\leq t\leq T}|{\bf y}( t)|\geq \varepsilon)\leq \nu.
\end{gather}
This implies that the solution ${\bf x}(t)\equiv 0$ is strongly stable as $t\leq \mu^{-2+\varkappa}$.

Now let us prove that the trivial solution of system \eqref{GS} is stable as $t\leq \mu^{-4+\varkappa}$. The Lyapunov function is constructed on the basis of the functions $V$ and $V_1$:
$$V_2({\bf y},t;T)=\big( V({\bf y},t)\big)^2+ \mu^2 a_1 V_1({\bf y},t;T),$$ where parameters $a_1$ and $T$ are defined later.
The action of the operator $\mathcal L$ on $V_2$ satisfies the inequality:
\begin{eqnarray*}
    \mathcal L V_2 & = & \, 2 V \,\mathcal L V + 2 \mu^2 \sum_{i,j=1}^n \sigma_{ij}\, \partial_{x_i}V\partial_{x_j}V + \mu^2 a_1 \mathcal L  V_1 \leq \\
    & \leq & - 2 \gamma V^2 + \mu^2 [2\, n^2 h\, (B+C) - a_1 \gamma ]   V
\end{eqnarray*}
as $|{\bf y}|\leq r_0$. We choose $a_1= 2\, n^2 h(B+C)\gamma^{-1}$, then $V_2\geq V^2\geq 0$ and $\mathcal LV_2\leq 0$ for all $({\bf y}, t)\in \mathcal I$.
Hence $V_2({\bf y}(s_ t),s_t;T)$ is a nonnegative supermartingale, and the following inequalities hold:
\begin{equation}\label{Est2}
\begin{array}{lll}
\displaystyle {\bf P}(\sup_{ 0\leq  t\leq T} |{\bf y}( t)|\geq \varepsilon)
    & = & \displaystyle {\bf P}(\sup_{ 0\leq  t\leq T} |{\bf y}( t)|^4\geq \varepsilon^4) \leq \\
    & \leq & \displaystyle {\bf P}(\sup_{ 0\leq  t\leq  T} \big( V({\bf y}( t), t)\big)^2\geq \varepsilon^4) \leq \\
    & \leq & \displaystyle {\bf P}(\sup_{ 0\leq  t\leq  T} V_2({\bf y}( t), t;T)\geq \varepsilon^4)  =  \\
    & = & \displaystyle {\bf P}(\sup_{ t\geq  0} V_2({\bf y}(s_ t),s_ t;T)\geq \varepsilon^4) \leq  \frac{V_2({\bf y}_0,0;T)}{\varepsilon^4}.
\end{array}
\end{equation}
At this step we choose $T=\mu^{-4+\varkappa}$ and define the parameters $\delta>0$, $\Delta>0$ such that
$$\frac{3A^2 \delta^4}{2}+ \frac{\Delta^4 a_1^2}{2}+ \Delta^{\varkappa}  a_1 n^2 h\, C \leq \varepsilon^4\nu.$$
Then the inequality $V_2({\bf y}_0,0; T)\leq\varepsilon^4 \nu$ holds for all $|{\bf y}_0|<\delta$, $\mu<\Delta$ and $G\in\mathcal A_h$.
Combining this with \eqref{Est2}, we get estimate \eqref{EstP} showing the stability on the time interval $[0; \mu^{-4+\varkappa}]$.

Finally let us prove the stability on the interval $0\leq t\leq \mu^{-2N+\varkappa}$ with $N\geq 3$. In this case we construct the perturbed Lyapunov function in the following form:
 \begin{eqnarray*}
        V_N({\bf y},t;T) & = & \big( V({\bf y},t)\big)^N+\mu^2 a_{N-1} V_{N-1}({\bf y},t;T), \\
        V_k({\bf y},t;T) & = & \big( V({\bf y},t)\big)^k+\mu^2 a_{k-1} V_{k-1}({\bf y},t;T), \quad k=2,\dots,N-1,
 \end{eqnarray*}
where the function $V_1({\bf y},t;T)$ is determined by formula \eqref{V1} and $T=\mu^{-2N+\varkappa}$.
The additional condition $\mathcal L V_k\leq 0$, imposed on each function $V_k$ in the ball $|{\bf y}|\leq r_0$, leads to the choice of the parameters $a_k$. Let $a_k=(k+1) n^2 h (B+C)\gamma^{-1}$, then the following inequalities hold:
 \begin{equation*}
    \begin{split}
  &      \mathcal L V_3\leq -3\gamma  V^3+2 \mu^2 \Big(3 n^2 h\, (B+C) - a_2 \gamma \Big)   V^2 \leq 0, \\
  &      \dots \\
   &     \mathcal L V_{N}\leq -N \gamma   V^{N}+(N-1) \mu^2  \Big(N n^2 h\, (B+C) - a_{N-1} \gamma \Big)   V^{N-1}\leq 0
    \end{split}
\end{equation*}
for all $|{\bf y}|\leq r_0$, $t\geq 0$, and $N=3,4,\dots$. Hence the function $V_N({\bf y}_0(s_t),s_t;T)$ is a nonnegative supermartingale. Since $V_N\geq V^N\geq 0$ for all $({\bf y}, t)\in \mathcal I$, we have the inequalities:
\begin{equation}\label{EstN}
\begin{array}{lll}
 \displaystyle {\bf P}(\sup_{ 0\leq  t\leq T} |{\bf y}( t)|\geq \varepsilon)
        & = & \displaystyle {\bf P}(\sup_{ 0\leq  t\leq  T} |{\bf y}( t)|^{2N}\geq \varepsilon^{2N})  \leq \\
        & \leq & \displaystyle {\bf P}(\sup_{ 0\leq  t\leq  T} \big(V({\bf y}( t), t)\big)^N\geq \varepsilon^{2N}) \leq \\
        & \leq & \displaystyle {\bf P}(\sup_{ 0\leq  t\leq  T} V_N({\bf y}( t), t;T)\geq \varepsilon^{2N})   = \\
        & = &  \displaystyle {\bf P}(\sup_{ t\geq  0} V_N({\bf y}(s_ t),s_ t;T)\geq \varepsilon^{2N}) \leq  \frac{V_N({\bf y}_0,0;T)}{\varepsilon^{2N}}.
\end{array}
\end{equation}
It can easily be checked that $V_N({\bf y}_0,0;T)=\mathcal O(|{\bf y}_0|^{2N})+\mathcal O(\mu^{\varkappa})$ as $|{\bf y}_0|\to 0$ and $\mu\to 0$. This yields that there exist $\delta>0$ and $\Delta>0$ such that $V_N({\bf y}_0,0;T)\leq \varepsilon^{2N} \nu$ for all $|{\bf y}_0|< \delta$,  $\mu<\Delta$. Combining this with \eqref{EstN}, we get estimate \eqref{EstP}. Therefore, $\forall \, N\in\mathbb N$, $h>0$ the solution ${\bf x}(t)\equiv 0$ to deterministic system \eqref{GS} is strongly stable with respect to white noise on the time interval $[0;\mu^{-2N+\varkappa}]$ uniformly for $G\in\mathcal A_h$. This completes the proof.
\end{proof}

\begin{Rem}
The theorem holds true if we choose $T=\mu^{-2N}\lambda (|{\bf y}_0|)$ with a positive continues function $\lambda(z)$ such that $\lambda(z)\to 0$ as $z \to 0$. In this case we get the inequality: $0\leq V_N({\bf y}_0,0;T)\leq m(|{\bf y}_0|^{2N}+\mu^{2N}+\lambda(|{\bf y}_0|))$ as $|{\bf y}_0|\leq  r_0$ and $\mu<1$ with a positive constant $m$.
The parameters $\delta$ and $\Delta$ are chosen such that  $\delta^{2N}+\lambda(\delta)\leq \varepsilon^{2N}\nu/(2m)$, $\Delta^{2N}\leq \varepsilon^{2N}\nu/(2m)$. Hence we have \eqref{EstP} for all $|{\bf y}_0|<\delta$ and $\mu<\Delta$. This implies that the solution ${\bf x}(t)\equiv 0$ is stable on the asymptotically long time interval $0 \leq t\leq \mu^{-2N}\lambda (|{\bf y}_0|)$.
\end{Rem}

\begin{Rem}
If the Lyapunov function $V({\bf x},t)$ have properties \eqref{LF} with $\gamma=0$, then the proposed construction of the perturbed Lyapunov function guarantees the stability of the trivial solution only on the time interval $0\leq t\leq o(\mu^{-2})$.
\end{Rem}

\section{Stability under damped perturbations}

Let us consider a narrower class $\mathcal J$ consisting of matrices $G({\bf y}, t)$ satisfying \eqref{Lip}, \eqref{Grow}, and having at least one continues function $\zeta(t)\in L_1(\mathbb R^+)$ such that
$$\sup_{|{\bf y}|\leq  r_0}\|\sigma({\bf y}, t)\|\leq \zeta( t) \quad \forall \,  t\geq  0, \quad \sigma=\frac{G\cdot G^\ast}{2}.$$
Let $\mathcal J_h$ be a set of matrices $G\in\mathcal J$ such that
$$\int\limits_{0}^\infty \zeta(\theta)\,d\theta\leq h.$$
We have
\begin{Th}\label{Th2}
Suppose that unperturbed system \eqref{GS} has a Lyapunov function $V({\bf x},t)$ satisfying \eqref{LF} with $\gamma=0$. Then for all $h>0$ the solution ${\bf x}(t)\equiv 0$ to system \eqref{GS} is strongly stable in probability with respect to white noise
as $t\geq 0$ uniformly for $G\in\mathcal J_h$.
\end{Th}
\begin{proof}
We fix parameters $h>0$, $\nu>0$, and $0<\varepsilon<r_0$. Suppose ${\bf y}( t)$ is a solution to system \eqref{GSP} with $|{\bf y}_0|<\delta$, $G\in\mathcal J_h$,  $\tau_{\varepsilon}$ is a first exit time of the solution ${\bf y}(t)$ from the domain $\mathcal B_\varepsilon=\{{\bf y}\in \mathbb R^n: |{\bf y}|<\varepsilon\}$, and  $s_t=\min\{ \tau_{\varepsilon}, t\}$. Then the function ${\bf y}(s_t)$ is the process stopped at the first exit time from the ball $\mathcal B_\varepsilon$.

The Lyapunov function for system \eqref{GSP} is constructed in the form
$$
    V_\mu({\bf y}, t)=V({\bf y}, t) + \mu^2 n^2  C \cdot \big( h - \int\limits_{ 0}^ t \zeta(\theta)\,d\theta\big).
$$
From the definition of the set $\mathcal J_h$ and properties of the function $V$ it follows that $V_\mu\geq V\geq 0$ and
$$
    \mathcal L V_\mu = \frac{dV}{d t}\Big|_{\eqref{GS}}  +  \mu^2\sum_{i,j=1}^n \sigma_{ij} \, \partial_{x_i}\partial_{x_j} V - \mu^2 n^2 C \,\zeta( t) \leq 0 \quad \forall\, |{\bf y}|\leq r_0, \quad  t\geq 0.
$$
Consequently, the function $V_\mu({\bf y}(s_ t),s_ t)$  is a nonnegative supermartingale, and we have the estimates:
\begin{equation}\label{EstI}
\begin{array}{lll}
    \displaystyle {\bf P}(\sup_{ t\geq  0} |{\bf y}( t)|\geq \varepsilon)
        & = & \displaystyle {\bf P}(\sup_{ t\geq   0} |{\bf y}(s_ t)|^2\geq \varepsilon^2)\leq \\
        & \leq & \displaystyle {\bf P}(\sup_{ t\geq  0} V({\bf y}(s_ t),s_ t)\geq \varepsilon^2) \leq \\
        & \leq & \displaystyle {\bf P}(\sup_{ t\geq  0} V_\mu({\bf y}(s_ t),s_ t)\geq \varepsilon^2) = \\
        & = & \displaystyle {\bf P}(\sup_{ t\geq  0} V_\mu({\bf y}(s_ t),s_ t)\geq \varepsilon^2) \leq   \frac{V_\mu({\bf y}_0,0)}{\varepsilon^2}.
\end{array}
\end{equation}
The last estimate follows from Doob's supermartingale inequality.
We choose $\delta>0$, $\Delta>0$ such that $A |{\bf y}_0|^2 +\mu^2 n^2 C h \leq \varepsilon^2 \nu$ for all $|{\bf y}_0|<\delta$, $\mu<\Delta$.
Hence, $V_\mu({\bf y}_0,0)\leq\varepsilon^2 \nu$. Combining this with \eqref{EstI}, we get estimate \eqref{EstP}. Therefore, the solution ${\bf x}(t)\equiv 0$ to system \eqref{GS} is strongly stable with respect to white noise as $t\geq 0$ uniformly for $G\in\mathcal J_h$.
\end{proof}

\section{Examples}
1. Let us consider the perturbed It\^{o} stochastic differential equation: $d y(t) = \mu d w(t)$. The solution $x(t)\equiv 0$ to the unperturbed system $\dot x=0$ is stable (but not asymptotically stable); the Lyapunov function $V=x^2$ satisfies estimates \eqref{LF} with $\gamma=0$.
Let us show that the trivial solution is strongly stable w.r.t. white noise (with $G\equiv 1$) on the time interval $0\leq t\leq o(\mu^{-2})$.

It can easily be checked that the function $y(t)=y_0+\mu w(t)$ is the solution to the perturbed equation. For simplicity, we assume $y_0=0$. Using the known formula for the Wiener process (see \cite[p.~173]{BS02}), we get
\begin{gather*}
    {\bf P}(\sup_{0\leq  t\leq T} |w(t)|\geq c)={\rm erfc}\Big(\frac{c}{\sqrt{2T}}\Big) + \frac{1}{\sqrt{2\pi T}}\sum\limits_{k\neq 0} (-1)^{k+1}\int\limits_{(2k-1)c}^{(2k+1)c} \exp\Big(-\frac{u^2}{2T}\Big) d u,
\end{gather*}
$c,T={\hbox{\rm const}}>0$.
This implies that for all $0<\varkappa<1$, $\varepsilon>0$, and $\nu>0$ there exists $\Delta=({\varepsilon^2}/{|4 \ln(\nu/\sqrt{8})|})^{1/\varkappa}$ such that
$${\bf P}(\sup_{0\leq  t\leq \mu^{-2+\varkappa}} |y(t)|\geq \varepsilon)\leq \nu$$
as $\mu<\Delta$.
In this example the strong stability is not preserved as $t\gg \mu^{-2}$. The stretching of the time interval can be achieved by imposing the additional constraints on the perturbation (see theorem~\ref{Th2}).

2. Let us consider a more complicated equation:
$$dy(t)=- \frac{y(1-y^2)}{1+y^2}\,d t+\mu\,G(y,t)\,dw( t), \quad 0<\mu\ll 1.$$
The corresponding deterministic system
 $$\frac{dx}{dt}=- \frac{x(1-x^2)}{1+x^2}$$
 has three equilibria: $\{-1, 0, 1\}$. It is easily shown that the trivial solution ${\bf x}(t)\equiv 0$ is locally exponentially stable; the Lyapunov function $V(x)=x^2$ satisfies the inequality $dV/d t\leq -V$ for all $|x|\leq 1/\sqrt3$, $t\geq 0$. From~\cite{Hasm65} it follows that in this case there is no stability of the trivial solution under white noise on a semi-infinite time interval with $G\in\mathcal A$.
However, it follows from theorem~\ref{Th1} that the solution $x( t)\equiv 0$ is strongly stable on the finite but asymptotically long time interval $0\leq  t\leq \mu^{-2N+\varkappa}$.

\section*{Acknowledgements}
The author is grateful to F.S. Nasyrov and L.A. Kalyakin for useful discussions and valuable comments.

This research was supported by the Russian Foundation for Basic Research  (project no. 14-01-31054).

\end{document}